\documentclass{article}

\usepackage{amsmath}
\usepackage{amsfonts}
\usepackage{amssymb}
\usepackage{amscd}
\usepackage{amsthm}
\usepackage{latexsym}
\usepackage{times}
\usepackage{mathrsfs}
\usepackage{epsfig}
\usepackage{color}

\usepackage{xy}

\def\B{\mathcal B}
\def\C{\mathbb C}
\def\d{\mathrm{d}}
\def\D{\mathcal D}
\def\H{{{\mathcal H}_{\!\Lambda}}}
\def\HH{\mathcal H}
\def\HR{{{\mathcal H}_{\!\R}}}
\def\Hrond{\mathfrak h}

\def\R{\mathbb R}

\def\TT{\mathscr T}

\def\K{\mathcal K}

\def\({\left(}
\def\){\right)}
\def\[{\left[}
\def\]{\right]}

\def\e{\mathop{\mathrm{e}}\nolimits}

\def\12{{\textstyle\frac12}}

\def\Pv{\mathrm{P.v.}}
\def\bv{\boldsymbol{\varphi}}

\newtheorem{Theorem}{Theorem}

\newtheorem{Assumption}[Theorem]{Assumption}
\newtheorem{Lemma}[Theorem]{Lemma}
\newtheorem{Corollary}[Theorem]{Corollary}
\newtheorem{Proposition}[Theorem]{Proposition}

\begin{document}

\title{On the wave operators
for the Friedrichs-Faddeev model}

\author{H. Isozaki and S. Richard\footnote{On leave from Universit\'e de Lyon; Universit\'e
Lyon 1; CNRS, UMR5208, Institut Camille Jordan, 43 blvd du 11 novembre 1918, F-69622
Villeurbanne-Cedex, France.
Supported by the Japan Society for the Promotion of Science (JSPS) and by
``Grants-in-Aid for scientific Research''.}}

\date{\small}
\maketitle \vspace{-1cm}

\begin{quote}
\emph{
\begin{itemize}
\item[] Graduate School of Pure and Applied Sciences,
University of Tsukuba, \\
1-1-1 Tennodai,
Tsukuba, Ibaraki 305-8571, Japan
\item[] \emph{E-mails:} isozakih@math.tsukuba.ac.jp, richard@math.univ-lyon1.fr
\end{itemize}
}
\end{quote}

\maketitle

\begin{abstract}
We provide new formulae for the wave operators in the context of the Friedrichs-Faddeev model. Continuity with respect to the energy of the scattering matrix and a few results on eigenfunctions corresponding to embedded eigenvalues are also derived.
\end{abstract}

\section{Introduction}\label{introduction}

In a series of recent works on scattering theory and Levinson's theorem \cite{KR1,KR3,KR5,KR6,RT} we advocate new formulae for the wave operators in the context of quantum scattering theory. Namely, let $H_0$ and $H$ be two self-adjoint operators in a Hilbert space $\HH$, and assume that $H_0$ has a purely absolutely continuous spectrum. In the time dependent framework of scattering theory, the wave operators $W_\pm$ are defined by the strong limits
\begin{equation*}
W_\pm := s-\lim_{t \to \pm \infty} \e^{itH} \e^{-itH_0}
\end{equation*}
whenever these limits exist. Then, our recent finding is that under suitable assumptions on $H_0$ and $H$ the following formula holds:
\begin{equation}\label{newf}
W_- =  1 + \bv(D)(S-1)+K
\end{equation}
where $S:=W_+^*W_-$ is the scattering operator, $D$ is an auxiliary  self-adjoint operator in $\HH$, $\bv$ is an explicit function and $K$ is a compact operator (we refer to Theorem \ref{thmmain} in Section \ref{secmain} for the precise statement). In other words the wave operator $W_-$ has, modulo compact operators, a very explicit and convenient form. Note that a similar formula for $W_+$ also exists.

For information, let us mention that \eqref{newf} was first proved with $K=0$ for Schr\"odinger operators with one $\delta$-interaction in space dimension $1$ to $3$ \cite{KR1}. This result was then fully extended to more regular potentials in the $1$-dimensional case \cite{KR5} and partially extended for the $3$-dimensional situation \cite{KR6}. In the article \cite{RT} the same formula was obtained for a rank-one perturbation, and in \cite{PR} the Aharonov-Bohm model was considered. Now, let us stress that the main difficulty for deriving \eqref{newf} relies on the proof of the compactness of the term $K$, and that this difficulty strongly depends on space dimensions. Indeed, even if in the context of potential scattering the $1$-dimensional problem is under control, the $3$-dimensional is much less tractable, and the even dimensional case has not been solved yet.

Our purpose in the present paper is to establish formula \eqref{newf} in the context of the Friedrichs-Faddeev model as presented in \cite[Sec.~4.1\&4.2]{Y}. In fact its interest is twofold: Firstly, embedded eigenvalues can exist in this model and they represent a special interest in our investigations. Secondly, the mentioned problem of space dimension is overtaken in this setting and does not play any role. Then, let us mention that an important corollary of formula \eqref{newf} is a straightforward proof of a topological version of Levinson's theorem once a suitable $C^*$-algebraic framework is introduced. However, since such a construction would not differ for this model from the ones already presented in \cite{KR5,KR6} and \cite{RT} we have decided not to go on here in that direction and to concentrate mainly on the derivation of \eqref{newf}.

Let us end this Introduction with a few references about this model. Already in 1938 Friedrichs proposed considering the pair of operators $(H_0,H_0+V)$ in $L^2([-1,1])$, where $H_0$ is the multiplication operator by the identity map and $V$ is an integral operator satisfying suitable conditions \cite{Fried}. The first important results on this problem were then proved by Faddeev in \cite{Fad}. Later on, the possible existence of singularly continuous spectrum for $H$ and the presence of embedded eigenvalues have attracted lots of attention, see for example \cite{DNY,L1,L2,PP}. Now, in Sections 4.1 and 4.2 of \cite{Y} a concise but rather complete presentation of the model is provided. Since our analysis is based on the results contained in this reference, we recall them in Section \ref{introFrie}. Our main contribution is then presented in Section \ref{secmain} while the two last sections are devoted to the proof of the compactness of the operator $K$ under two different sets of assumptions, see Propositions \ref{easy} and \ref{vp}.
Let us finally mention that the continuity with respect to the energy of the scattering matrix is a by-product of our analysis, and that a few results on eigenfunctions corresponding to embedded eigenvalues are also derived.

\section{Framework}\label{introFrie}

In this section, we introduce the Friedrichs-Faddeev model as presented in Sections 4.1 and 4.2 of \cite{Y} and recall a few results. Let $\Lambda:=[a,b]\subset \R$ be a finite interval and let $\Hrond$ be a Hilbert space with norm  $\|\cdot\|_\Hrond$ and scalar product $\langle \cdot,\cdot \rangle_\Hrond$. We denote by $\H$ the Hilbert space $L^2(\Lambda;\Hrond)$, and consider in $\H$ the bounded self-adjoint operator $H_0$ acting on $f \in C(\Lambda;\Hrond)\subset \H$ as $[H_0 f](\lambda):= \lambda f(\lambda)$ for any $\lambda \in \Lambda$.

Now, let $v : \Lambda \times \Lambda\to \K(\Hrond)$ be a H\"older continuous function of exponent $\alpha_0\in (1/2,1]$ which takes values in the algebra $\K(\Hrond)$ of compact operators on $\Hrond$. More precisely, we assume that $v(\lambda,\mu)\in \K(\Hrond)$ for all $\lambda, \mu \in \Lambda$ and that
\begin{equation*}
\sup_{\lambda, \mu \in \Lambda}
\|v(\lambda,\mu)\|_{\B(\Hrond)} + \sup_{\lambda,\mu,\lambda',\mu'\in \Lambda}
\big\|v(\lambda',\mu')-v(\lambda,\mu)\big\|_{\B(\Hrond)}\Big/
\big(|\lambda -\lambda'| + |\mu'-\mu|\big)^{\alpha_0} <\infty \ .
\end{equation*}
We also require that $v(\lambda,\mu)= v(\mu,\lambda)^*$ and that  the function $v$ vanishes at the boundary of its domain, {\it i.e.}~for all $\lambda$ and $\mu$
$$
v(\lambda,a)=v(\lambda,b)= v(a,\mu)=v(b,\mu)=0\ .
$$

It then follows from these assumptions that the operator $H:=H_0+V$, with $V$ defined on $f\in C(\Lambda;\Hrond)\subset \H$ and for $\lambda \in \Lambda$ by
\begin{equation*}
[Vf](\lambda) := \int_\Lambda v(\lambda,\mu)\;\!f(\mu)\;\!\d \mu\ ,
\end{equation*}
is a bounded and self-adjoint operator in $\H$.
In fact, $V$ is a compact perturbation of $H_0$. It is then a standard result that the essential spectra of $H_0$ and $H$ coincide with $\Lambda$. Furthermore, it is proved in \cite[Sec.~4.1 \& 4.2]{Y} that $H$ has no singularly continuous spectrum and that the point spectrum $\sigma_p(H)$ of $H$ is exhausted by a finite set of eigenvalues of finite multiplicities.

Now, for $z \in \C\setminus \R$ let us set $R_0(z):=(H_0-z)^{-1}$ and $R(z):=(H-z)^{-1}$ for the resolvents of $H_0$ and $H$, respectively. For suitable $z\in \C$ we also set
\begin{equation}\label{defTz}
T(z):=V-V\;\!R(z)\;\!V\ .
\end{equation}
Clearly, $T(\cdot)$ is an operator-valued meromorphic function in $\C\setminus \Lambda$ and has poles only at points of the discrete spectrum on $H$.
Additional properties of this operator are recalled in the next proposition. We refer to \cite[Thm.~4.1.1]{Y} for its proof and for more detailed properties of $T(z)$. In the sequel $\Pi$ denotes the closed complex plane with a cut along the spectrum $\Lambda$ of the operator $H_0$. Note also that "integral operator" means here operator-valued integral operator.

\begin{Proposition}\label{Yaf1}
For $z \in \Pi\setminus \sigma_p(H)$, the operator $T(z)$ is an integral operator which kernel $t(\cdot,\cdot,z): \Lambda\times \Lambda \to \K(\Hrond)$ satisfies
\begin{equation*}
\|t(\lambda',\mu',z')-t(\lambda,\mu,z)\|_{\B(\Hrond)}\leq {\rm c} \;\!
\big(|\lambda'-\lambda|+|\mu'-\mu|+|z'-z|\big)^{\alpha}
\end{equation*}
for any $\alpha <\alpha_0$, any $\lambda,\mu,\lambda',\mu' \in \Lambda$ and
any $z,z' \in \Pi\setminus \sigma_p(H)$.
The constant ${\rm c}$ is independent of the variables $z,z'$ outside arbitrary small neighbourhoods of $\sigma_p(H)$. Furthermore, on the boundary of $\Lambda \times \Lambda$ the kernel $t(\cdot,\cdot,z)$ vanishes.
\end{Proposition}

Based on the analysis of the operator $T(z)$, a proof of the existence and of the asymptotic completeness of the wave operators is proposed in \cite[Sec.~4.2]{Y}. More precisely, under the mentioned hypotheses on $v$ the wave operators $W_{\pm}$
exist, are isometries and their ranges are equal to $\H_{\!,p}(H)^\bot$. Here $\H_{\!,p}(H)$ denotes the subspace of $\H$ spanned by the eigenfunctions of $H$.
Let us now set $\kappa(H):=\big(\sigma_p(H)\cap \Lambda\big) \cup \{a,b\}$, which corresponds to the set of embedded eigenvalues together with the thresholds $a$ and $b$. Then, on the dense subset $\D$ of $\H$ defined by $\D:=C^{\infty}_c\big(\Lambda \setminus \kappa(H);\Hrond\big)$ the following stationary representations hold:
\begin{equation}\label{defW}
[W_{\pm}f](\lambda)= f(\lambda) - \int_\Lambda t(\lambda,\mu,\mu\mp i 0)\;\!(\lambda - \mu \pm i0)^{-1}\;\!f(\mu)\;\!\d \mu \qquad \forall f \in \D, \lambda \in \Lambda\ .
\end{equation}
The precise meaning of the second term on the r.h.s.~is the following: one first considers the family of expressions
$$
\big[\TT_\pm(\varepsilon,\tau)f\big](\lambda):=
\int_\Lambda t(\lambda,\mu,\mu\mp i \varepsilon)\;\!(\lambda - \mu \pm i\tau)^{-1}\;\!f(\mu)\;\!\d \mu
$$
for $\varepsilon,\tau>0$.
Then, the second term $[\TT_\pm f](\lambda)$ in \eqref{defW} is obtained by taking the strong limit, {\it i.e.}
$$
\TT_\pm f: s-\lim_{\varepsilon \searrow 0, \tau \searrow 0} \TT_\pm(\varepsilon,\tau)f \qquad \forall f \in \D.
$$
We refer to \cite[Sec.~4.2.2]{Y} for a justification of these stationary formulas.

Similarly, the scattering operator $S:=W_+^*\;\!W_-$ can also be expressed in terms of the kernel of $T(z)$.
More precisely, the scattering operator is an operator-valued multiplication operator, {\it i.e.}~$[Sf](\lambda)= s(\lambda)f(\lambda)$ for almost every $\lambda \in \Lambda$, and the scattering matrix $s(\lambda)\in \B(\Hrond)$ is given for $\lambda \in \Lambda\setminus \sigma_p(H)$ by
\begin{equation}\label{slambda}
s(\lambda)= 1-2\pi i \;\!t(\lambda,\lambda,\lambda+i0)\ .
\end{equation}
We also mention that for $\lambda \in \Lambda\setminus \sigma_p(H)$ the operator $s(\lambda)$ is unitary, that $s(\lambda)-1\in \K(\Hrond)$ and that the map
\begin{equation*}
\Lambda \setminus \sigma_p(H) \ni \lambda \mapsto s(\lambda)\in \B(\Hrond)
\end{equation*}
is H\"older continuous in norm for any exponent $\alpha <\alpha_0$.

Let us finally derive a new expression for the wave operators, concentrating on $W_-$ since a similar formula for $W_+$ can then be deduced. So for any $f\in \D$ we consider the equalities:
\begin{eqnarray}\label{f1}
\nonumber [(W_--1)f](\lambda)&=& -\int_\Lambda t(\lambda,\mu,\mu + i 0 )\;\!(\lambda - \mu - i0)^{-1}\;\!f(\mu)\;\!\d \mu \\
\nonumber &=& -\int_\Lambda (\lambda - \mu - i0)^{-1}\;\! t(\mu,\mu,\mu + i 0 )\;\!f(\mu)\;\!\d \mu \\
\nonumber &&-\int_\Lambda (\lambda - \mu - i0)^{-1}\;\! \big[t(\lambda,\mu,\mu+i0)-t(\mu,\mu,\mu + i 0 )\big]\;\!f(\mu)\;\!\d \mu \\
&=& \frac{1}{2\pi i} \int_\Lambda (\lambda - \mu - i0)^{-1}\;\! \big[s(\mu)-1\big]\;\!f(\mu)\;\!\d \mu + [Kf](\lambda)
\end{eqnarray}
with $[Kf](\lambda):= \int_\Lambda k(\lambda,\mu) f(\mu)\;\!\d \mu\ $ and
\begin{equation}\label{defdeK}
k(\lambda,\mu):=- (\lambda - \mu - i0)^{-1}\;\! \big[t(\lambda,\mu,\mu+i0)-t(\mu,\mu,\mu + i 0 )\big].
\end{equation}

Our goal is now twofold: firstly one seeks for a simpler expression for the first term in \eqref{f1}, and secondly one looks for sufficient conditions which would imply the compactness of the operator $K$.

\section{In the rescaled energy's representation}\label{secmain}

In this section we derive a simpler expression for the first term in \eqref{f1} by working in another representation of the original Hilbert space. The following construction is inspired by \cite{BSB} from which we borrow the idea of rescaled energy's representation.

Let us consider the Hilbert space $\HR:=L^2(\R;\Hrond)$ and the unitary map $U:\H\to \HR$ defined on any $f \in C(\Lambda;\Hrond)\subset\H$ and for $x \in \R$ by
\begin{equation*}
[Uf](x):=\sqrt{\frac{b-a}{2}}\frac{1}{\cosh(x)}f\Big(
\frac{a+b\e^{2x}}{1+\e^{2x}}\Big)\ .
\end{equation*}
The inverse of this map is given for $\varphi \in C_c(\R;\Hrond)\subset \HR$ and $\lambda \in \Lambda$ by
\begin{equation*}
[U^{-1}\varphi](\lambda) = \sqrt{\frac{b-a}{2}}\frac{1}{\sqrt{(\lambda -a)(b-\lambda)}} \varphi\Big(\frac{1}{2}\ln\frac{\lambda-a}{b-\lambda}\Big)\ .
\end{equation*}
Then, let $M$ be an operator-valued multiplication operator in $\H$ by a function $m \in L^\infty\big(\Lambda;\B(\Hrond)\big)$.
A straightforward computation leads to the following expression for its representation in $\HR$: $\widetilde M:= UMU^{-1}$ is the operator-valued multiplication operator by the function $\widetilde m(\cdot)=m\big(\frac{a+b\e^{2\cdot}}{1+\e^{2\cdot}}\big)$.
In particular, by choosing $m(\lambda)= \lambda$ one obtains that $UH_0U^{-1}$ is the operator of multiplication by the bounded function
$\widetilde h_0$ defined by $\widetilde h_0(x) =\frac{a+b\e^{2x}}{1+\e^{2x}}$. Note that this function is strictly increasing on $\R$ and takes the asymptotic values $\widetilde h_0(-\infty)=a$ and $\widetilde h_0(\infty) = b$.

Let us now concentrate on the singular part of the first term in \eqref{f1}. More precisely, for any $f \in C^{\infty}_c\big(\Lambda;\Hrond\big)$ we concentrate on the expression
\begin{equation*}
[Tf](\lambda):=\frac{1}{2\pi i} \int_\Lambda (\lambda - \mu - i0)^{-1} \;\!f(\mu)\;\!\d \mu
\end{equation*}
which is equal to
\begin{equation*}
\frac{1}{2\pi i} \;\Pv\!\!\int_\Lambda (\lambda - \mu)^{-1} \;\!f(\mu)\;\!\d \mu
+ \frac{1}{2} f(\lambda).
\end{equation*}
A straightforward computation leads then to the following equality for any $\varphi \in C^{\infty}_c\big(\R;\Hrond\big)$ and $x \in \R$:
\begin{equation*}
[UTU^{-1}\varphi](x)=\frac{1}{2}\Big[
\frac{i}{\pi } \;\Pv\!\!\int_\R \frac{\varphi(y)}{\sinh(y-x)}\;\d y + \varphi(x)
\Big]\ .
\end{equation*}

Thus, if $X$ and $D$ denote respectively the usual self-adjoint operators in $\HR$ corresponding to the formal expressions $[X\varphi](x)=x\varphi(x)$ and $[D\varphi](x)=-i\varphi'(x)$, then one is led to the equality
\begin{equation*}
[UTU^{-1}\varphi](x)= \frac{1}{2}\Big[\frac{i}{\pi }  \;\Pv\!\! \int_\R \frac{1}{\sinh(y)}\;[\e^{iyD}\varphi]\;\d y + \varphi\Big](x) \ .
\end{equation*}
Furthermore, by taking into account the formula
\begin{equation*}
\frac{i}{\pi}\;\Pv\!\! \int_\R\frac{\e^{-ixy}}{\sinh(y)}\;\d y= \tanh\big(\frac{\pi}{2}x\big)
\end{equation*}
one finally obtains
$UTU^{-1} =
\frac{1}{2}\big\{1-\tanh\big(\frac{\pi}{2}D\big)\big\}$.
By collecting these results and by a density argument, one has thus proved :

\begin{Theorem}\label{thmmain}
The following equality holds:
$$
U(W_--1)U^{-1} =\frac{1}{2}\Big\{1-\tanh\big(\frac{\pi}{2}D\big)\Big\} \big(\widetilde S-1\big) + \widetilde K
$$
with $\widetilde S = USU^{-1}$ and $\widetilde K= UKU^{-1}$.
The operator $\widetilde S$ is equal to the operator-valued multiplication operator defined by the function $\R\ni x\mapsto s\big(\frac{a+b\e^{2x}}{1+\e^{2x}}\big)\in \B(\Hrond)$ for almost every $x \in \R$.
\end{Theorem}

By taking the the asymptotic completeness into account, one easily deduces from the relation $W_+=W_-S^*$ the following corollary.

\begin{Corollary}
The following equality holds:
$$
U(W_+-1)U^{-1} =\frac{1}{2}\Big\{1+\tanh\big(\frac{\pi}{2}D\big)\Big\} \big(\widetilde S^*-1\big) + \widetilde K \widetilde S^*
$$
with $\widetilde S^* = US^*U^{-1}$.
\end{Corollary}

\section{Compactness, the easy case}

In this section we show that, with an implicit condition, the compactness of the operator $K$ defined in \eqref{f1} is easily checked. In the following section, this implicit assumption will be removed.

So let us assume that the point spectrum of $H$ inside $\Lambda$ is empty, namely $\sigma_p(H)\cap [a,b]=\emptyset$. In such a situation, Proposition \ref{Yaf1} can be strengthened in the sense that the H\"older continuity holds for all $z,z' \in \Pi$ and that the constant ${\rm c}$ can be chosen independently of $z$ and $z'$. It then follows that the kernel $k$ introduced in \eqref{defdeK} corresponds to a compact-valued Hilbert-Schmidt operator. Indeed, one has
\begin{align*}
& \int_\Lambda \int_\Lambda \big\|k(\lambda,\mu) \big\|_{\B(\Hrond)}^2\;\!\d \lambda\;\!\d \mu \\
&=  \int_\Lambda \int_\Lambda |\lambda - \mu|^{-2}\big\| t(\lambda,\mu,\mu+i0)-t(\mu,\mu,\mu + i0) \big\|_{\B(\Hrond)}^2 \;\!\d \lambda\;\!\d \mu \\
&\leq {\rm c}^2 \int_\Lambda \int_\Lambda |\lambda -\mu|^{2(\alpha -1)} \;\!\d \lambda\;\!\d \mu \\
&<\infty
\end{align*}
since one can choose the exponent $\alpha \in (1/2,\alpha_0)$. In other words:

\begin{Proposition}\label{easy}
In the framework introduced in Section \ref{introFrie} and with the additional assumption $\sigma_p(H)\cap [a,b]=\emptyset$, the operator $K$ defined in \eqref{f1} belongs to $\K(\H)$.
\end{Proposition}

Let us also note that with this implicit assumption, the map $\Lambda \ni\lambda \mapsto s(\lambda)\in \B(\Hrond)$ is H\"older continuous in norm for any exponent $\alpha <\alpha_0$.

\section{Compactness, the general case}

In this section we do not assume that the point spectrum inside $\Lambda$ is empty. However, we shall impose a stronger regularity to the kernel of the operator $V$ in order to ensure the compactness of $K$.

First of all, as mentioned in Section~\ref{introFrie} there is only a finite number of embedded eigenvalues and each one is of finite multiplicity.
So, let us denote by $\{\lambda_n\}_{n=1}^N\subset \Lambda$ the finite set of embedded eigenvalues, repeated accordingly to multiplicity, and let $\{f_n\}_{n=1}^N \subset \H$ be a family of corresponding mutually orthogonal eigenfunctions. Without loss of generality, we assume that each $f_n$ is of norm $1$. The one-dimensional orthogonal projection on $f_n$ is denoted by $|f_n\rangle \langle f_n|$.

Instead of directly studying the kernel defined in \eqref{defdeK} we shall come back to its original form in terms of the operator $T(z)$ defined in \eqref{defTz}. More precisely, let us consider the kernel
\begin{align}
\nonumber &-(\lambda - \mu - i0)^{-1} \Big\{
\big[V-V\;\!R(\mu + i0)\;\!V\big](\lambda,\mu) -
\big[V-V\;\!R(\mu + i0)\;\!V\big](\mu,\mu)
\Big\} \\
\label{terme1} &=
- (\lambda - \mu - i0)^{-1} \big\{v(\lambda,\mu)-v(\mu,\mu)\big\}  \\
\label{terme2} &\qquad +
(\lambda - \mu - i0)^{-1} \Big\{
\big[V\;\!R(\mu + i0)\;\!V\big](\lambda,\mu) -
\big[V\;\!R(\mu + i0)\;\!V\big](\mu,\mu)
\Big\}\ .
\end{align}
Clearly, by taking the H\"older continuity of the kernel of $V$ into account one has
\begin{equation*}
\int_\Lambda \int_\Lambda |\lambda -\mu|^{-2} \big\|v(\lambda,\mu)-v(\mu,\mu)\big\|^2_{\B(\Hrond)}
\d \lambda \;\!\d \mu
\leq {\rm Const.}~\int_\Lambda \int_\Lambda
|\lambda -\mu|^{2(\alpha_0-1)} \d \lambda \;\!\d \mu<\infty\ .
\end{equation*}
It follows that the operator corresponding to the kernel \eqref{terme1} is again a compact-valued Hilbert-Schmidt operator.

For the second term we shall consider the following decomposition $1 = P + \sum_{n=1}^N |f_n\rangle \langle f_n|$, with $P:= 1-\sum_{n=1}^N |f_n\rangle \langle f_n|$, which is going to be introduced on the right of the factors $R(\mu + i0)$ in \eqref{terme2}. First, some easy computations lead to the following inequalities:
\begin{align*}
&\int_\Lambda \int_\Lambda |\lambda -\mu|^{-2}\Big\|
\big[V\;\!R(\mu + i0)|f_n\rangle \langle f_n|\;\!V\big](\lambda,\mu) -\big[V\;\!R(\mu + i0)|f_n\rangle \langle f_n|\;\!V\big](\mu,\mu)
\Big\|_{\B(\Hrond)}^2\d \lambda \;\!\d \mu \\
&=\int_\Lambda \int_\Lambda |\lambda -\mu|^{-2}\Big\|
\big[ \frac{1}{\lambda_n-\mu-i0}\int_\Lambda \;\!\big[v(\lambda,\nu)-v(\mu,\nu)\big] |f_n(\nu)\rangle_\Hrond {}_\Hrond\!\langle [Vf_n](\mu)| \d \nu
\Big\|_{\B(\Hrond)}^2\d \lambda \;\!\d \mu \\
&\leq {\rm Const.}~\|f_n\|^2_\H \int_\Lambda \int_\Lambda |\lambda-\mu|^{2(\alpha_0-1)}\big\|\big[R_0(\lambda_n +i0)Vf_n\big](\mu)\big\|_\Hrond^2\;\!\d \lambda \;\!\d \mu\ .
\end{align*}
Since $\alpha_0>1/2$ the estimate $\sup_{\mu\in \Lambda}\int_\Lambda|\lambda -\mu|^{2(\alpha_0-1)}\;\!\d \lambda<\infty$ is satisfied. Thus the above expression is finite if $R_0(\lambda_n+i0)Vf_n$ belongs to $\H$. However, we shall show subsequently that $-R_0(\lambda_n \pm i 0)Vf_n=f_n$, which clearly justifies the claim.

Now, we shall concentrate on proving the following inequality:
\begin{equation}\label{avoir}
\int_\Lambda \int_\Lambda |\lambda -\mu|^{-2}
\Big\|
\big[V\;\!R(\mu + i0)P\;\!V\big](\lambda,\mu) -
\big[V\;\!R(\mu + i0)P\;\!V\big](\mu,\mu)
\Big\|_{\B(\Hrond)}^2\d \lambda \;\!\d \mu < \infty.
\end{equation}
This will be obtained by imposing a stronger regularity to the function $v$. At the end, by collecting these various results it will prove that the term $K$ is equal to a finite sum of compact-valued Hilbert-Schmidt operators.

So let us study of the operator $VR(z)PV$ for $z$ in the upper half complex plane. For that purpose and following \cite[Sec.~4.1]{Y}, we introduce for any $\alpha \in (0,1]$ the Banach space $C^\alpha_0(\Lambda;\Hrond)$ of $\Hrond$-valued H\"older continuous functions vanishing at $a$ and $b$ endowed with the norm
$$
\|f\|_\alpha := \sup_{\lambda,\lambda' \in \Lambda} \Big(\|f(\lambda)\|_\Hrond+ \|f(\lambda)-f(\lambda')\|_\Hrond \Big/  |\lambda - \lambda' |^{\alpha}\Big) .
$$
This space is not separable. We therefore define the Banach space  $\dot{C}^\alpha_0(\Lambda;\Hrond)$ as the closure of $C^\infty_0(\Lambda;\Hrond)$ with the above norm.
Clearly the inclusion $\dot{C}^\alpha_0(\Lambda;\Hrond) \subset C^\alpha_0(\Lambda;\Hrond)$ holds, but one also has $C^{\alpha_1}_0(\Lambda;\Hrond) \subset \dot{C}^{\alpha_2}_0(\Lambda;\Hrond)$ if $\alpha_2<\alpha_1\leq 1$.

Now, let us define $A(z):=-VR_0(z)$. It is proved in \cite[Lem.~4.1.2]{Y} that for any $z \in \Pi$ and any $\alpha_j<\alpha_0$ the operator $A(z)$ is compact from $C^{\alpha_1}_0(\Lambda;\Hrond)$ to $\dot{C}^{\alpha_2}_0(\Lambda;\Hrond)$. In particular, it follows from this and from the Fredholm alternative for Banach spaces that the operator $1-A(z)$ is invertible in the space $\dot{C}^{\alpha}_0(\Lambda;\Hrond)$ for any $\alpha \in \big(0,\alpha_0)$ whenever the equation $A(z)f=f$ has no nontrivial solution. Equivalently, this corresponds to the fact that $\ker\big(1-A(z)\big)=\{0\}$ in $\dot{C}^{\alpha}_0(\Lambda;\Hrond)$.
So, in order to study the operator
$$
R(z)P=R_0(z)\big(1-A(z)\big)^{-1}P
$$
in a suitable space, one needs to get a better understanding of the operator $P= 1-\sum_{n=1}^N |f_n\rangle \langle f_n|$. This is the content of the next lemma and its corollary. Its assumption is highlighted before the statement.
We refer to \cite{GJ,MW} and to references mentioned therein for related statements on the regularity of eigenfunctions corresponding to embedded eigenvalues.

\begin{Assumption}\label{stronger}
For each $\mu \in \Lambda$ the map
$
\Lambda \ni \lambda \mapsto v(\lambda,\mu)\in \K(\Hrond)
$
is norm-differentiable, with derivative denoted by $v'(\lambda,\mu)$, and the map $v' : \Lambda \times \Lambda\to \K(\Hrond)$ is a H\"older continuous function of exponent $\alpha_0'\in (0,1]$. Furthermore $v'(a,\mu)=v'(b,\mu)=0$ for arbitrary $\mu \in \Lambda$.
\end{Assumption}

Clearly, if $v$ satisfies this assumption it also satisfies the original regularity condition with $\alpha_0=1$. In the sequel, we shall tacitly take this fact into account. For any $\lambda \in \Lambda$ and any $\varepsilon>0$, we set
$$
o_\varepsilon(\lambda):=(\lambda-\varepsilon,\lambda+\varepsilon)\cap \Lambda\ .
$$

\begin{Lemma}\label{regul}
Suppose that Assumption \ref{stronger} holds for some $\alpha_0'\in (0,1]$.
Then for each eigenvalue $\lambda_n \in \Lambda$ of $H$ the corresponding eigenfunction $f_n$ belongs to $C^{\alpha_0'}_0(\Lambda;\Hrond)$.
\end{Lemma}

\begin{proof}
Assume first that $f_n \in \H$ satisfies $Hf_n = \lambda_n f_n$ for some $\lambda_n \in (a,b)$. It implies that for almost every $\lambda \in \Lambda\setminus\{\lambda_n\}$ one has $f_n(\lambda)=-\frac{1}{\lambda -\lambda_n}[Vf_n](\lambda)$. Then,  since the r.h.s.~is well defined for every $\lambda \in \Lambda \setminus\{\lambda_n\}$ and $Vf_n\in C^1_0(\Lambda;\Hrond)$, one infers in particular that $f_n \in C^{\alpha_0'}_0(\Lambda \setminus o_\varepsilon(\lambda_n);\Hrond)$ for any $\varepsilon>0$. In other words, one can choose a representative element of $f_n \in \H$ in $C^{\alpha_0'}_0(\Lambda \setminus o_\varepsilon(\lambda_n);\Hrond)$.
Furthermore, one also infers that the property $[Vf_n](\lambda_n)=0$ holds, and by taking then the regularity condition on $v$ into account it follows that
\begin{equation}\label{vendredi}
[Vf_n](\lambda)= [Vf_n](\lambda) - [Vf_n](\lambda_n)
=(\lambda - \lambda_n) \int_0^1 [V'f_n]\big(\lambda_n + s(\lambda - \lambda_n)\big)\d s
\end{equation}
with $[V'f_n](\lambda) =\int_\Lambda v'(\lambda,\mu)f_n(\mu) \d \mu$. By inserting the r.h.s.~of \eqref{vendredi} in the equality  $f_n(\lambda)=-\frac{1}{\lambda -\lambda_n}[Vf_n](\lambda)$ one deduces that $f_n \in C^{\alpha_0'}_0(\Lambda;\Hrond)$ and that $f_n(\lambda_n)= -[V'f_n](\lambda_n)$. This proves the statement for $\lambda_n \in (a,b)$. Finally, the special case $\lambda_n\in \{a,b\}$ is proved  similarly by taking the additional condition $v'(a,\mu)=v'(b,\mu)=0$ into account.
\end{proof}

One easily deduces from this lemma the following consequence on the operator $P$:

\begin{Corollary}\label{corol1}
Let Assumption \ref{stronger} hold for some $\alpha_0'\in (0,1]$. Then
the operator $P$ belongs to  $\B\big(\dot{C}^{\alpha}_0(\Lambda;\Hrond)\big)$ for any $\alpha<\alpha_0'$.
\end{Corollary}

Now, in order to study the operator
\begin{equation}\label{explicit}
VR(z)PV= VR_0(z)\big(1-A(z)\big)^{-1}PV = -A(z) \big(1-A(z)\big)^{-1}PV,
\end{equation}
we shall suppose that Assumption \ref{stronger} holds and consider $\alpha\in (0,\alpha_0')$.
Recall first that $V$ maps $\H$ into $C^{1}_0(\Lambda;\Hrond)\subset \dot{C}^{\alpha}_0(\Lambda;\Hrond)$, and that $P$ maps $\dot{C}^{\alpha}_0(\Lambda;\Hrond)$ into itself.
Thus it is natural to consider the operator $\big(1-A(z)\big)^{-1}$ on $\dot{C}^{\alpha}_0(\Lambda;\Hrond)$.
For that purpose, we recall that the solutions of the equation $A(z)f=f$ in $\dot{C}^{\alpha}_0(\Lambda;\Hrond)$ are in one-to-one relation with the eigenfunctions of the operator $H$ \cite[Lem.~4.1.4]{Y}. More precisely, if $f$ is an eigenfunction of $H$ associated with the eigenvalue $\lambda \in \sigma_p(H)$, then $g:=-Vf$ satisfies $A(\lambda\pm i0)g=g$ (solution for both signs simultaneously). Alternatively, if $g$ is a solution of the equation $A(\lambda \pm i 0)g=g$ for some $\lambda \in \R$, then $g(\lambda)=0$ and $R_0(\lambda \pm i 0)g$  is an eigenfunction of $H$ associated with the eigenvalue $\lambda$, and thus $\lambda\in \sigma_p(H)$.
These relations imply in particular that $-R_0(\lambda_n \pm i 0)V f_n = f_n$, or in other words $f_n$ is an eigenfunction of the operator $-R_0(\lambda_n \pm i 0)V$ associated with the eigenvalue $1$.

Now, it follows from \cite[Lem.~4.1.2]{Y} that for $z \in \Pi$ the operator $A(z)$ belongs to $\K\big(\dot{C}^{\alpha}_0(\Lambda;\Hrond)\big)$, and
for $z \in \Pi \setminus \sigma_p(H)$ the operator $\big(1-A(z)\big)^{-1}$ is an element of $\B\big(\dot{C}^{\alpha}_0(\Lambda;\Hrond)\big)$.
Furthermore, the maps
\begin{equation*}
\Pi \ni z \mapsto
A(z) \in \K\big(\dot{C}^{\alpha}_0(\Lambda;\Hrond)\big)
\end{equation*}
and
\begin{equation*}
\Pi \setminus \sigma_p(H) \ni z \mapsto
\big(1-A(z)\big)^{-1} \in \B\big(\dot{C}^{\alpha}_0(\Lambda;\Hrond)\big)
\end{equation*}
are norm-continuous. We shall show that this latter result can be extended for all $z \in \Pi_d:=\big(\Pi \setminus \sigma_p(H)\big)\cup \{\lambda_1, \dots, \lambda_n\}$ once the operator $P$ is applied on the right of the operator  $\big(1-A(z)\big)^{-1}$ as in \eqref{explicit}. Note that $\Pi_d$ is equal to $\Pi$ with the discrete spectrum of $H$ excluded

\begin{Lemma}\label{continu}
Suppose that Assumption \ref{stronger} holds for some $\alpha_0'\in (0,1]$.
Then for each $\alpha \in (0,\alpha_0')$ and  each $z \in \Pi_d$ the operator $\big(1-A(z)\big)^{-1}P$ is well defined and bounded on $\dot{C}^{\alpha}_0(\Lambda;\Hrond)$. Furthermore, the map
\begin{equation*}
\Pi_d \ni z \mapsto
\big(1-A(z)\big)^{-1}P \in \B\big(\dot{C}^{\alpha}_0(\Lambda;\Hrond)\big)
\end{equation*}
is norm-continuous.
\end{Lemma}

\begin{proof}
Clearly, we can concentrate on the neighbourhood of a singular point $\lambda_n\in \Lambda$ and consider only the limit from above (the limit from below is completely similar). For that purpose, let us simply set $T(\lambda_n):=\big(1-A(\lambda_n+i0)\big)\in \B\big(\dot{C}^{\alpha}_0(\Lambda;\Hrond)\big)$. By the Fredholm alternative for Banach spaces, the equation $T(\lambda_n)f=0$ with $f \in \dot{C}^{\alpha}_0(\Lambda;\Hrond)$ and the equation $T(\lambda_n)^*F=0$ in the adjoint space $\dot{C}^{\alpha}_0(\Lambda;\Hrond)^*$ have the same finite number $m$ of linearly independent solutions, which we denote respectively by $f^j$ and $F^j$. Furthermore, the equation $T(\lambda_n)f=g$ has a solution for a given $g \in \dot{C}^{\alpha}_0(\Lambda;\Hrond)$ if and only if $F^j(g)=0$. In such a case, the solution is obtained by $f = T(\lambda_n)^{-1}g$.

Now, since $A(\lambda_n +i0)^* = -R_0(\lambda_n - i0)V$, one observes that the elements $F^j$ are nothing but linearly independent elements of the subspace of $\H$ generated by the eigenfunctions of $A(\lambda_n+i0)^*$ associated with the  eigenvalue $1$ (which means that $m$ is equal to the multiplicity of the eigenvalue $\lambda_n$). Furthermore, the equation $F^j(g)=0$ reduces to $\langle F^j,g\rangle_\H=0$. Then, by choosing $g \in P\dot{C}^{\alpha}_0(\Lambda;\Hrond)\subset \dot{C}^{\alpha}_0(\Lambda;\Hrond)$, one clearly has $F^j \bot g$, which means that the condition $F^j(g)=0$ is satisfied. One concludes that on the set $P\dot{C}^{\alpha}_0(\Lambda;\Hrond)$ the operator $T(\lambda_n)$ has a bounded inverse. Finally, the continuity follows from a straightforward argument, see for example of proof of \cite[Lem.~1.8.1]{Y}.
\end{proof}

Before proving the main result of this section, let us observe that for $z \in \C \setminus \R$ the operator $VR(z)P$ is an integral operator. Indeed, from the relation $VR(z) = T(z)R_0(z)$ and since $T(z)$ is an integral operator, one infers that $VR(z)$ is an integral operator. Then, multiplying this operator by $P= 1-\sum_{n=1}^N |f_n\rangle \langle f_n|$ does not change this property.

\begin{Proposition}\label{vp}
Suppose that Assumption \ref{stronger} holds for some $\alpha_0' \in (1/2,1]$.
Then the inequality \eqref{avoir} is satisfied, and thus the term $K$ is a finite sum of compact-valued Hilbert-Schmidt operators.
\end{Proposition}

\begin{proof}
Let us fix $\alpha$ with $1/2<\alpha <\alpha_0'$ and for $z \in \Pi_d$ we set $C(z):= -A(z) \big(1-A(z)\big)^{-1}P$ which belongs to $\K\big(\dot{C}^{\alpha}_0(\Lambda;\Hrond)\big)$, as a consequence of the previous results.
Now, the study of the l.h.s.~of \eqref{avoir} leads naturally to the analysis of the kernel of the operator $C(z)V$.
For that purpose, one easily observes that for $\zeta \in \Hrond$ and fixed $\mu\in \Lambda$, the map $\Lambda \ni \lambda \mapsto v(\lambda,\mu)\zeta \in \Hrond$ belongs to $C^{1}_0(\Lambda;\Hrond)$, or stated differently $v(\cdot,\mu)\zeta \in C^{1}_0(\Lambda;\Hrond)$. In particular, it implies that $v(\cdot,\mu)\zeta \in \dot{C}^{\alpha}_0(\Lambda;\Hrond)$. It then follows from the above observation on $C(z)$ that $C(z)v(\cdot,\mu)\zeta \in \dot{C}^{\alpha}_0(\Lambda;\Hrond)$ and
\begin{equation}\label{maj}
\|C(z)v(\cdot,\mu)\zeta\|_{\alpha}
\leq {\rm Const.}~\|v(\cdot,\mu)\zeta\|_{\alpha}
\leq {\rm Const.}~\|\zeta\|_\Hrond.
\end{equation}
Note that the constants can be chosen independently of $\mu\in \Lambda$ and of $z$ belonging to a compact subset of $\Pi_d$.

One then infers from this inequality and from the equalities
$$
[VR(z)PV](\lambda,\mu)\zeta =  [C(z)V](\lambda,\mu)\zeta = [C(z)v(\cdot,\mu)\zeta](\lambda)
$$
that
\begin{align*}
&|\lambda -\mu|^{-\alpha}
\Big\|
\big[V\;\!R(\mu + i0)P\;\!V\big](\lambda,\mu) -
\big[V\;\!R(\mu + i0)P\;\!V\big](\mu,\mu)
\Big\|_{\B(\Hrond)} \\
&= \sup_{\zeta \in \Hrond, \|\zeta\|_{\Hrond}=1}
|\lambda -\mu|^{-\alpha} \Big\|
\big[V\;\!R(\mu + i0)P\;\!V\big](\lambda,\mu)\zeta -
\big[V\;\!R(\mu + i0)P\;\!V\big](\mu,\mu)\zeta
\Big\|_{\Hrond} \\
&= \sup_{\zeta \in \Hrond, \|\zeta\|_{\Hrond}=1}
|\lambda -\mu|^{-\alpha}\Big\|\big[
[C(\mu+i0)v(\cdot,\mu)\zeta](\lambda)-[C(\mu+i0)v(\cdot,\mu)\zeta](\mu)
\big]\Big\|_\Hrond \\
&\leq  \sup_{\zeta \in \Hrond, \|\zeta\|_{\Hrond}=1} \big\|C(\mu+i0)v(\cdot,\mu)\zeta\big\|_{\alpha} \\
&\leq  {\rm Const.}
\end{align*}
with the constants independent of $\mu$ and $\lambda$.
Inserting this estimate into \eqref{avoir} leads directly to the result.
\end{proof}

Let us finally prove a result on the continuity of the map $\lambda\mapsto s(\lambda)$ under a similar assumption.

\begin{Proposition}
Suppose that Assumption \ref{stronger} holds for some $\alpha_0' \in (0,1]$. Then the map
\begin{equation*}
\Lambda \ni \lambda\mapsto s(\lambda) \in \B(\Hrond)
\end{equation*}
is norm-continuous.
\end{Proposition}

\begin{proof}
It clearly follows from relation \eqref{slambda} and the properties stated below it that the statement would be proved if one shows that the map $\Lambda \setminus \sigma_p(H)\ni \lambda \mapsto s(\lambda)\in \B(\Hrond)$ can be continuously extended on all $\Lambda$. In particular, it is sufficient to show that the map $\lambda \mapsto t(\lambda, \lambda, \lambda +i0)\in \B(\Hrond)$ is norm-continuous on $\Lambda$.
For that purpose, let us recall that $t(\lambda,\lambda,z) = v(\lambda,\lambda) - [VR(z)V](\lambda,\lambda)$. Since the first term easily satisfies the necessary continuity property, we shall concentrate on the second term.

As for previous computations we consider the decomposition $1 = P + \sum_{n=1}^N |f_n\rangle \langle f_n|$, which is introduced on the right of the factor $R(z)$ in the expression $VR(z)V$. We first consider the operator $VR(z)|f_n\rangle \langle f_n|V$. One observes that
\begin{align*}
\big[VR(\lambda + i 0)|f_n\rangle \langle f_n|V\big](\lambda,\lambda) &=
\frac{1}{\lambda_n-\lambda -i0}\big[|Vf_n\rangle \langle Vf_n|\big](\lambda,\lambda)\\
&=  -|[Vf_n](\lambda)\rangle_\Hrond {}_\Hrond\!\langle [R_0(\lambda_n+i0)Vf_n](\lambda)| \\
&= |[Vf_n](\lambda)\rangle_\Hrond {}_\Hrond\!\langle f_n(\lambda)|.
\end{align*}
The continuity of the map $\Lambda \ni \lambda \mapsto \big[VR(\lambda + i 0)|f_n\rangle \langle f_n|V\big](\lambda,\lambda)\in \K(\Hrond)$ follows then from the property $f_n \in C^{\alpha_0'}_0(\Lambda;\Hrond)$, which is a consequence of Lemma \ref{regul}, and from the regularity of $v$.

Let us now consider the term $VR(\lambda+i0)PV = C(\lambda+i0)V$ with the operator $C(z)$ introduced and studied in the proof of Proposition \ref{vp}. In fact, part of the following arguments are based on results obtained in that proof. We fix $\alpha \in (0,\alpha_0')$ and first observe that for $\zeta \in \Hrond$ and $\lambda,\lambda'\in \Lambda$ one has
\begin{align}\label{aetudier}
\nonumber &[C(\lambda'+i0)V](\lambda',\lambda')\zeta - [C(\lambda+i0)V](\lambda,\lambda)\zeta \\
\nonumber &= \big[C(\lambda'+i0)V-C(\lambda+i0)V\big](\lambda',\lambda')\zeta
+ \Big\{[C(\lambda+i0)V](\lambda',\lambda')- [C(\lambda+i0)V](\lambda,\lambda')\Big\}\zeta \\
&\qquad + \Big\{[C(\lambda+i0)V](\lambda,\lambda')- [C(\lambda+i0)V](\lambda,\lambda)\Big\}\zeta.
\end{align}
We shall study separately each term and show that their norms vanish (independently of $\zeta$) as $\lambda' \to \lambda$. For the first one, we have
\begin{align*}
&\big\|\big[C(\lambda'+i0)V-C(\lambda+i0)V\big] (\lambda',\lambda')\zeta\big\|_\Hrond \\
& =\big\|
\big[\big\{C(\lambda'+i0)-C(\lambda+i0)\big\}v(\cdot,\lambda')\zeta\big](\lambda')
\big\|_\Hrond \\
&\leq \big\|
\big\{C(\lambda'+i0)-C(\lambda+i0)\big\}v(\cdot,\lambda')\zeta
\big\|_\alpha \\
&\leq \big\|C(\lambda'+i0)-C(\lambda+i0)\big\|_{\B(\dot{C}^{\alpha}_0(\Lambda;\Hrond))} \
\|v(\cdot,\lambda')\zeta\|_\alpha \\
& \leq {\rm Const.}~ \big\|C(\lambda'+i0)-C(\lambda+i0)\big\|_{\B(\dot{C}^{\alpha}_0(\Lambda;\Hrond))} \|\zeta\|_\Hrond
\end{align*}
with a constant independent of $\lambda$ and $\lambda'$. The continuity of the map $\Lambda \ni \lambda \mapsto C(\lambda+i0)\in \K(\Hrond)$ gives the necessary continuity.

For the second term of \eqref{aetudier}, one simply has to recall that $C(\lambda+i0)v(\cdot,\lambda)\zeta$ belongs to $\dot{C}^{\alpha}_0(\Lambda;\Hrond)$, and then
\begin{align*}
&\big\|\big\{[C(\lambda+i0)V](\lambda',\lambda')- [C(\lambda+i0)V](\lambda,\lambda')\big\}\zeta\big\|_\Hrond \\
&=\big\|\big[C(\lambda+i0)v(\cdot,\lambda')\zeta\big](\lambda') - \big[C(\lambda+i0)v(\cdot,\lambda')\zeta\big](\lambda)
\big\|_\Hrond\\
&\leq |\lambda'-\lambda|^\alpha
\big\| C(\lambda+i0)v(\cdot,\lambda')\zeta \big\|_\alpha\\
& \leq {\rm Const.}~|\lambda'-\lambda|^\alpha\;\|\zeta\|_\Hrond,
\end{align*}
with the last inequality based on \eqref{maj}. Again, the necessary continuity follows from these inequalities.
Finally, the last term can be treated similarly by taking the relation $[C(\lambda+i0)V](\lambda,\mu) = \big\{[C(\lambda-i0)V](\mu,\lambda)\big\}^*$ into account and by expressing the norm with a scalar product.
\end{proof}



\begin{thebibliography}{00}

\bibitem{BSB} J. Bellissard, H. Schulz-Baldes, \emph{Scattering theory for lattice operators in dimension $d\geq 3$}, in prepartion.

\bibitem{DNY} E.M. Dyn'kin, S.A. Naboko, S.I Yakovlev, \emph{The boundary of finiteness of the singular spectrum in the selfadjoint Friedrichs model}, St. Petersburg Math. J. {\bf 3} (1992) no. 2, 299--313.

\bibitem{Fad} L.D. Faddeev, \emph{On a model of Friedrichs in the theory of perturbations of the continuous spectrum}, Amer. Math. Soc. Transl. (2) {\bf 62} (1967), 177--203.

\bibitem{Fried} K. Friedrichs, \emph{\"Uber die Spektralzerlegung eines Integraloperators}, Math. Ann. {\bf 115} (1938) no. 1, 249--272.

\bibitem{GJ} S. Gol\'enia, T. Jecko, \emph{A new look at Mourre's commutator theory}, Complex Anal. Oper. Theory {\bf 1} (2007) no. 3, 399--422.

\bibitem{KR1} J. Kellendonk, S. Richard, \emph{Levinson's theorem for Schr\"odinger operators with point interaction: a topological approach}  J. Phys. A {\bf 39} (2006) no. 46, 14397--14403.

\bibitem{KR3} J. Kellendonk, S. Richard,
\emph{The topological meaning of Levinson's theorem, half-bound states included},
J. Phys. A: Math. Theor. {\bf 41} (2008), 295207.

\bibitem{KR5} J. Kellendonk, S. Richard,
\emph{On the structure of the wave operators in one dimensional potential scattering}, Mathematical Physics Electronic Journal {\bf 14} (2008), 1--21.

\bibitem{KR6} J. Kellendonk, S. Richard, \emph{On the wave operators and Levinson's theorem
for potential scattering in $\R^3$}, to appear in Asian-European Journal of Mathematics.

\bibitem{L1} S.N. Lakajev,
\emph{Discrete spectrum of operator valued Friedrichs models},
Comment. Math. Univ. Carolin. {\bf 27} (1986) no. 2, 341--357.

\bibitem{L2} S.N. Lakajev, \emph{A result about imbedded eigenvalues in the operator valued Friedrichs model},
Comment. Math. Univ. Carolin. {\bf 27} (1986) no. 3, 479--490.

\bibitem{MW} J.S. Moeller, M. Westrich, \emph{Regularity of eigenstates in regular Mourre theory}, J. Funct. Anal. {\bf 260} (2011), 852--878.

\bibitem{PR} K. Pankrashkin, S. Richard, \emph{Spectral and scattering theory for the Aharonov-Bohm operators},
Rev. Math. Phys. {\bf 23} (2011), 53--81.

\bibitem{PP} B.S. Pavlov, S.V. Petras, Pavlov, \emph{The singular spectrum of a weakly perturbed multiplication operator}, Functional Anal. Appl. {\bf 4} (1970), 136--142.

\bibitem{RT} S. Richard, R. Tiedra de Aldecoa, \emph{New formulae for the wave operators for a rank one interaction},
Integral Equations and Operator Theory {\bf 66} (2010), 283--292.

\bibitem{Y} D.R. Yafaev, \emph{Mathematical scattering theory. General theory},
Translations of Mathematical Monographs {\bf 105}, American Mathematical Society, Providence,
RI, 1992.
\end{thebibliography}
\end{document}